\documentclass[runningheads]{llncs}
\usepackage[T1]{fontenc}
\usepackage{graphicx}
\usepackage{amsmath,amssymb}
\usepackage[algosection,ruled,lined,linesnumbered,longend]{algorithm2e}
\usepackage{hyperref}
\usepackage{cleveref}

\usepackage{amsthm}
\usepackage{pifont,bm,tikz}
\usetikzlibrary{arrows.meta, positioning, calc}
\usepackage{textcomp,extarrows,graphfig}
\usepackage{boxedminipage,float}
\usepackage{subcaption,caption}

\usepackage{color,colortbl}

\usepackage{mdframed,pifont}
\usepackage{booktabs}
\usepackage{textcomp}
\usepackage{mathabx}

\usepackage{verbatim}    
\usepackage[title]{appendix}
\usepackage[misc]{ifsym} 

\begin{document}
	\title{Publicly Verifiable Private Information Retrieval Protocols Based on Function Secret Sharing}
	\titlerunning{PVPIR on FSS}
	
	\author{
		Lin Zhu\inst{1} \and
		Lingwei Kong\inst{1} \and
		Xin Ning\inst{2} \and
		Xiaoyang Qu\inst{1} \and
		Jianzong Wang\inst{1}\textsuperscript{(\Letter)}
	}
	
	\authorrunning{Zhu et al.}
	\institute{
		Ping An Technology (Shenzhen) Co., Ltd., Shenzhen, China \\
		\email{zhulinneko@gmail.com, 530124621@qq.com, quxiaoy@gmail.com} \\
		\email{\Letter jzwang@188.com}
		\and
		Shenzhen Youfang Information Technology Co., Ltd., Shenzhen, China \\
		\email{ningxin@163.com}
	}

	\maketitle              
	
	\begin{abstract}		
		Private Information Retrieval (PIR) is a fundamental cryptographic primitive that enables users to retrieve data from a database without revealing which item is being accessed, thereby preserving query privacy. 
		However, PIR protocols also face the challenge of result verifiability, as users expect the reconstructed data to be trustworthy and authentic.
		In this work, we propose two effective constructions of publicly verifiable PIR (PVPIR) in the multi-server setting, which achieve query privacy, correctness, and verifiability simultaneously.
		We further present three concrete instantiations based on these constructions.
		For the point query, our protocol introduces minimal computational overhead and achieves strong verifiability guarantees with significantly lower communication costs compared to existing Merkle tree–based approaches. 
		For the predicate query, the communication complexity of our scheme remains stable as the database size increases,
		demonstrating strong scalability and suitability for large-scale private query applications.

	\end{abstract}

	\section{Introduction}	
	{\em Private Information Retrieval} (PIR) \cite{CGKS95} aims to provide the user with the means to retrieve data from a database without the database owner or keeper learning any information about the particular item that was queried.
	The PIR user wants to download the data ${\bm x}_i$ from database ${\bm X}=({\bm x}_1, {\bm x}_2, \cdots, {\bm x}_N)$ without leaking the corresponding index $i$ to any servers.
	PIR is a fundamental building block for many proposed privacy-sensitive applications in the literature, including 
	collective authority \cite{STVWJGGKF16}
	signature-producing blockchain \cite{KEPNLIJB17},
	private database search \cite{WYSVM17},
	private video streaming consumption \cite{GCMSAW16}
	and credential breach reporting \cite{KM20}.
	In these services, any user can query the data, and the datasets themselves are not sensitive, 
	but privacy concerns can make the user unwilling to disclose his or her interest information.

	Most of the existing classic PIR protocols \cite{LP23,AMHSV23,DG15,MAEW23,SRR14,C04,CMS99,ACLS18,MCR21,CK20}
	consider honest-but-curious servers, with the assumption that servers strictly follow the protocols' specifications.
	A malicious server may arbitrarily deviate from the protocol, potentially causing the user to receive incorrect data. In modern cloud environments, where servers can be powerful but untrusted, PIR protocols resilient to such behavior are highly desirable. This motivates verifiable PIR (VPIR) protocols \cite{ZSN14,SYO22}, which enable the client to verify the correctness of retrieved data, detect misbehavior, and ensure integrity. Such guarantees are critical in privacy-sensitive applications, including public key directories, block-chain records, and federated analysis, where incorrect responses may cause serious failures or security breaches. Table~\ref{tab:comparison} summarizes several advanced PIR schemes and their verifiability.

	While verifiable PIR protocols improve classical PIR by ensuring server response integrity, most existing schemes offer only private verifiability, allowing only the querying client to check correctness. This limits transparency and prevents external auditing. In contrast, public verifiability enables any third party with the public verification key to verify the consistency of the reconstructed message with the database—i.e., detect malicious behavior—without revealing the query or its result.
	
	We explore the publicly verifiable PIR (PVPIR) protocols in the $k$-server setting, leveraging \emph{function secret sharing (FSS)}, 
	and focus on both \emph{predicate queries} and \emph{point queries}.
	Our constructions are motivated by efficiency, public verifiability, and scalability.
	Compared to privately verifiable PIR, where only the querying client can check the correctness of the response, publicly verifiable PIR provides stronger auditability and transparency by allowing any third party to verify the result, independent of the client. 
	This property is particularly valuable in settings where trust must be externally certified, such as delegated access, compliance auditing, or multi-stakeholder data sharing. 
	On the other hand, in both the privately or publicly verifiable case, 
	achieving verifiability often opens the door to select failure attacks, 
	wherein an adversary deliberately induces failure by returns selective incorrect answers, and then exploits the client’s observable reaction to infer the secret query message.
	This attack model is particularly powerful because it does not rely on breaking cryptographic primitives directly, but rather on exploiting error-handling behavior or side-channel feedback in implementation. 
	As a result, protocols that offer strong theoretical security may still leak critical information in real-world deployments if they fail to account for selective failure vectors.

	Existing Merkle tree-based protocols do not address both challenges simultaneously, leaving a crucial gap in the design of secure and trustworthy PIR systems. 
	Our work bridges this gap by introducing the first PIR framework that simultaneously achieves public verifiability and provable robustness against select failure attacks. 
	This contribution advances the practical deployment of PIR in sensitive domains where verifiability, privacy, and robustness against active adversaries are all essential.

	\vspace{-1em}
	\subsection{Our Contributions}
	The primary contribution of this work is the design of three publicly verifiable PIR constructions in the two-server model. These protocols achieve both query privacy and results integrity, while supporting predicate and point queries efficiently. 
	Our first two schemes handle general predicate queries and provide scalable public verifiability under the discrete logarithm(DL) and RSA assumptions respectively.
	The third scheme is actually an instantiation of the first predicate query construction on the point function domain. It achieves public verifiability, while the Merkle tree-based scheme does not.	
	
	For any integer $k\geq 1$, our $k$-server PVPIR model
	for function class ${\mathcal F}\subseteq {\sf Funs}[[N]\times \{0,1\}^{\ell}, \mathbb{F}]$ 
	has a database ${\bm X}=({\bm x}_1,\ldots,{\bm x}_N)$, weights ${\bm \omega}\in \mathbb{F}^N$, 
	a {\em user} and $k$ {\em  servers} $\mathcal{S}_1, \mathcal{S}_2, \ldots, \mathcal{S}_k$.
	Each server keeps an identical copy of the encoding $X$.
	Our PVPIR model consists of four algorithms $\Pi = ({\sf KeyGen}, {\sf Query}, {\sf Answer},\\
	{\sf Reconstruct})$.
	To split the query function, the user will firstly run $\Pi.{\sf KeyGen}$ to generate a public key ${\sf pk}$ and a secret key ${\sf sk}$, and then run $\Pi.{\sf Query}$ to convert the query function $f$ into $k$ shares $({\bm q}_1, \ldots, {\bm q}_k)$ and a public verification key ${\sf vk}$.
	For every $j\in[k]$, the server $\mathcal{S}_j$ runs $\Pi.{\sf KeyGen}$ and uses $({\sf pk}, {\bm X}, {\bm \omega}, {\bm q}_j)$ to generate an answer ${\bm a}_j$ to the user.
	Finally, the user makes use of $\Pi.{\sf Reconstruct}(\{{\bm a}_j\}_{j=1}^k, {\sf pk}, {\sf vk})$ to locally reconstruct and verify the reconstructed message.
	A PVPIR protocol is {\em $t$-private} if any $t$ out of $k$ servers cannot learn the user’s query $f$, even when observing the client’s output. This notion is strictly stronger than classical PIR. The protocol is {\em secure} against active adversaries if the probability that the user accepts an incorrect result is negligible.

	According to the descriptions above, we know the public verifiability comes from 
	lightweight cryptography. 
	In contrast to prior FSS-based APIR schemes, which are limited to be private verification, our three constructions support \emph{public verifiability}. 
	This allows external auditors or third parties to validate results without access to any secret keys, 
	which enables broader transparency and stronger integrity guarantees in multi-party scenarios.

	\vspace{-8mm}
	\begin{table}
		\caption{\small Comparison of our constructions and other state-of-art PIR schemes, with $N$ being the size of the database.}
		\vspace{0.5mm}
		\centering
		\begin{tabular}[thb]{l|l|l|l|l}
			\toprule
			Scheme &  Servers & Verifiable & Communication & Computation \\
			\midrule
			Tree-PIR\cite{LP23} & 2 & \ding{55} & $\mathcal{O}(\sqrt{N})$ & $\mathcal{O}(\sqrt{N}\log^2N)$\\
			VPIR\cite{ZSN14} & 2 & \cellcolor{cyan!25}\ding{51} & $\mathcal{O}(N^{1/3}\log N)$ & $\backslash$ \\
			APIR(Merkle tree, point query)\cite{CNCWF23} & 2 & \cellcolor{cyan!25}\ding{51} & $\mathcal{O}(\sqrt{N}\log N)$ & $\mathcal{O}(N\log N)$\\
			Crust\cite{AH23} & 2 & \cellcolor{cyan!25}\ding{51} & $\mathcal{O}(\sqrt{N})$ & $\mathcal{O}(\sqrt{N})$\\
			\rowcolor{blue!25}Our work (scheme $\Pi_3$) & 2 & \cellcolor{cyan!25}\ding{51}(publicly) & $\mathcal{O}(\lambda\log N)$ & $\mathcal{O}(N)$\\
			\hline
			DPF-PIR\cite{GI14} & 2 & \ding{55} & 
			$\mathcal{O}(\log N)$ & $\mathcal{O}(N)$ \\
			
			APIR(FSS-based)\cite{CNCWF23} & 2 & \cellcolor{cyan!25}\ding{51}(privately) & $\mathcal{O}(\lambda\log N)$ & $\mathcal{O}(N)$\\
			\rowcolor{blue!25}Our work (scheme $\Pi_1$) & 2 & \cellcolor{cyan!25}\ding{51}(publicly) & $\mathcal{O}(\lambda\log N)$ & $\mathcal{O}(N)$\\
			\rowcolor{blue!25}Our work (scheme $\Pi_2$) & 2 & \cellcolor{cyan!25}\ding{51}(publicly) & $\mathcal{O}(\lambda\log N)$ & $\mathcal{O}(N)$\\
			\bottomrule
		\end{tabular}
		\label{tab:comparison}
	\end{table}
	\vspace{-1em}

 	All three protocols achieve communication complexity as low as $O(\lambda \log N)$, where $N$ is the database size and $\lambda$ is the security parameter. 
 	Through both theoretical analysis and empirical evaluation, we demonstrate that our constructions offer superior scalability and efficiency compared to state-of-the-art baselines such as APIR and Crust, as shown in Table~\ref{tab:comparison}.
 	Our schemes can also resist the selective failure attack, ensuring the privacy of the query even exposing the verification result to the adversary.
 	Our schemes also remain secure when at most $k-1$ servers are malicious. 
 	Even in the presence of an adversarial server, the client can either retrieve the correct result or detect misbehavior with overwhelming probability.

	\subsection{Related Work}  
	The notion of public verifiability has been extensively studied across multiple branches of modern cryptography, with the central goal of enabling any third party to independently verify the correctness of a cryptographic object or protocol execution.
	This idea was later extended to publicly verifiable secret sharing (PVSS) schemes
	such as those of Feldman\cite{FP87} and Pedersen\cite{PT92}, 
	which allow external parties to check the consistency of distributed shares in threshold protocols. 
	Similarly, in the context of zero-knowledge proofs, 
	non-interactive zero-knowledge (NIZK) systems\cite{BFM88,JA07}
	and modern zk-SNARK\cite{G13,BCTV14,G16}
	constructions provide succinct and publicly verifiable proofs of knowledge, forming the backbone of verifiable computation and blockchain applications.
	
	Public verifiability has also been a core design principle in verifiable computation and proofs of retrievability/storage, where outsourced servers must convince anyone that the results or stored data are correct without requiring trust.

	More recently, verifiable delay functions (VDFs)\cite{BBBF18} have been employed to provide publicly verifiable randomness and fairness in decentralized settings. 
	These developments highlight the broad utility of public verifiability for ensuring transparency, auditability, and non-repudiation in distributed systems.
	
	Building on these foundations, publicly verifiable private information retrieval (PIR) adapts the same paradigm to database queries: in addition to preserving query privacy, the server must return responses whose correctness can be validated by anyone (not only the querying client). 
	This property makes publicly verifiable PIR particularly appealing in decentralized or blockchain-based scenarios, where transparency and accountability are as important as privacy.

	\subsection{Applications} 
	Our multi-server verifiable schemes many be applied into scenarios that require privacy-preserving and verifiable computation of polynomial or arithmetic circuits on sensitive query.
	Public verifiability verifiable PIR enables not only the querying user but also any third party to independently verify the correctness of the retrieved result. 
	Compared with privately verifiable schemes, 
	which restricts verification to the user alone, public verifiability offers stronger auditability and transparency. 
	This property is particularly valuable in scenarios where results need to be shared, delegated, 
	or externally audited without compromising trust or correctness.
	
	Typical application scenarios include block-chain-based systems, where users query on-chain or off-chain data and attach publicly verifiable proofs to ensure trust-less validation; 
	lightweight clients or IoT devices, which may delegate verification to more capable parties; 
	and collaborative data access settings, where multiple parties rely on a shared result with verifiable integrity. 
	Public verifiability also provides a strong foundation for non-repudiation and regulatory compliance in data-sensitive environments.

	\section{Preliminaries}
	\label{sec:pre}
	
	\subsection{Notation}
	\label{sec:notation}
	We use $\mathbb{N}$ to denote the set of natural numbers, and let $\mathbb{F}$ denote a finite field, with $|\mathbb{F}|$ representing the number of elements in $\mathbb{F}$.
	For any $N \in \mathbb{N}$, we write $[N] := {1, 2, \ldots, N}$.
	Given a group $\mathbb{G}$, we denote its identity element by $1_{\mathbb{G}}$.
	Let $\mathbb{F}^n$ denote the $n$-dimensional vector space over the field $\mathbb{F}$.
	For finite sets $S$ and $T$, we write $\mathsf{Funs}[S, T]$ to denote the set of all functions from domain $S$ to co-domain $T$.
	
	We say that a function $\epsilon(\lambda)$ is {\em negligible} in $\lambda$ and denote
	$\epsilon(\lambda)={\sf negl}(\lambda)$, if for any $c>0$, there exists
	$\lambda_c>0$ such that $\epsilon(\lambda)<\lambda^{-c}$ for all $\lambda \ge \lambda_c$.
	We denote a function family ${\mathcal F}$ as an infinite collection of functions, and each function $f$ involves two efficient procedures {\sf IdentifyDomain} and {\sf Evaluate}.
	The former $D_f \leftarrow {\sf IdentifyDomain}(1^\lambda, f)$ maps the function $f$ into its domain space $D_f$,
	while the latter $f(x) \leftarrow {\sf Evaluate}(f, x)$ means the output of $f$ at the input $x\in D_f$.
	The symbol $\perp$ is an output that indicates rejection.

	\subsection{\bf Function Secret Sharing}
	\label{subsec:fss}
	The function secret sharing in multi-server setting PIR protocols enables a client to compute the image of the database under function $f$, 
	without revealing any information to these servers.
	The FSS scheme supports to split the function $f$ into several keys, which enables each server to compute a secret share of $f(x)$.

	\begin{definition}
	(Function secret sharing, FSS.)
	A $k$-party function secret sharing\cite{BGI15} protocol is defined with respect to a function class $\mathcal{F}$.
	Each function $f\in \mathcal{F}$ maps elements in some input space to a finite group or field $\mathbb{F}$. 
	Along with an additive output decoding algorithm ${\tt Dec}^{+}(f_1, f_2, \ldots, f_k) = \sum_{i=1}^{k} g_i$ (computing the sum of shares with group operator $+$), 
	an additive function secret protocol consists of a pair of polynomial-time algorithms ${\sf FSS}=({\sf Gen}, {\sf Eval})$ with the following syntax:		
	\begin{itemize}		
		\item ${\sf Gen}(1^\lambda, f) \rightarrow (f_1, f_2, \ldots, f_k)$: On input the security parameter $1^\lambda$ and function description $f\in \mathcal{F}$, the key generation algorithm outputs $k$ function-secret-shares ($f_1, f_2, \ldots, f_k$). 
		
		\item ${\sf Eval}(i, f_i, x) \rightarrow f_i(x)$: On input a party index $i$, share $f_i$ and input string $x$, the evaluation algorithm outputs a value $f_i(x)$.
	\end{itemize}
	
	\end{definition}
	
	Its correctness guarantees that applying all the secret share of $f(x)$ can recover the target evaluation $f(x)$.
	Its security ensures that set of any $k-1$ keys does not divulge any information of the function $f$.
	\begin{itemize}
		\item[$\bullet$] {\bf Correctness}: For all $f\in \mathcal{F}, x\in D_f$, 
		$$
		\Pr \left[
		\begin{aligned}
			(f_1, \cdots, f_k) \leftarrow 
			& \mathsf{Gen}(1^{\lambda}, f): \\
			&{\tt Dec}^{+}\left(
			\mathsf{Eval}(1, f_1, x), \cdots, \mathsf{Eval}(k, f_k, x) \right)
			= f(x)
		\end{aligned}
		\right ] = 1
		$$
		
		\item[$\bullet$] {\bf Security}: 
		Consider the following experiment of indistinguishability challenge ${\sf EXP}_{\mathcal{A}, {\sf FSS}}(T, \lambda)$ in {Fig.}~\ref{fig:sec of FSS}, for corrupted parties $T\subset [p]$, $|T|= k-1$
		we say the scheme $(\mathsf{Gen}, \mathsf{Eval})$ is secure if for all non-uniform PPT adversaries $\mathcal{A}$, it holds that 
		the advantage of $\mathcal{A}$ in guessing $b$ in the above experiment, 
		$\mathsf{Adv}(1^\lambda, \mathcal{A}) \leq {\sf negl}(\lambda)$, 
		where  $\mathsf{Adv}(1^\lambda, \mathcal{A}) := \left|\Pr[b=b']-1/2\right|$ 
		and the probability is taken over the randomness of the challenger and of $\mathcal{A}$.
	\end{itemize}
	
	\vspace{-1em}
	\begin{figure}[ht]
		\begin{center}
			\begin{boxedminipage}{11cm}
				\begin{enumerate}
					\item[(a)] 
					The adversary ${\mathcal A}$ outputs $(f^0, f^1, \mathsf{aux}) \leftarrow \mathcal{A}(1^\lambda)$, where $f^0, f^1\in \mathcal{F}$ with the same domain.
					
					\item[(b)] 
					The challenger samples $b\leftarrow \{0,1\}$ and $(f_1, f_2, \cdots, f_k) \leftarrow \mathsf{Gen}(1^\lambda, f^b)$.
					
					\item[(c)]
					The adversary outputs a guess $b' \leftarrow \mathcal{A}(\{f_i\}_{i\in T}, \mathsf{aux})$, given the keys for corrupted $T$.
				\end{enumerate}
			\end{boxedminipage}
			\caption{\small Experiment  
				${\sf EXP}_{\mathcal{A}, {\sf FSS}}(T, \lambda)$}
			\label{fig:sec of FSS}
			\vspace{-2em}
		\end{center}		
	\end{figure}
	
	\vspace{-1em}
	\begin{definition}{\bf (Function class closed under scalar multiplication).} 
		\label{def: closed under scalar multiplication}
		Let $\mathbb{F}$ be a class of functions whose co-domain is a
		finite field $\mathbb{F}$. 
		Then we say that the function class $\mathbb{F}$ is closed
		under scalar multiplication, if for all functions $f\in \mathcal{F}$ and for all scalars $\alpha\in \mathbb{F}$, it holds that the function $\alpha \cdot f \in \mathcal{F}$.
	\end{definition}

	\section{\bf The Publicly Verifiable PIR Model}
	\label{sec:pvpir model}
	In this section, we formalize the notion of \emph{Publicly Verifiable Private Information Retrieval (PVPIR)}. While classic PIR protocols ensure query privacy, they typically lack integrity guarantees, leaving users vulnerable to incorrect or tampered responses from malicious servers. 
	To address this limitation, we extend the PIR model by incorporating explicit verification mechanisms that allow the correctness of the server's response to be checked, even by other third parties.
	
	We define the standard syntax of a PVPIR protocol, where the client has function $f\in \mathcal F$, the server  database $\bm X$ and weights ${\bm \omega}$, 
	it consists of key generation, query, answering and reconstruction procedures, i.e, $({\sf KeyGen}, {\sf Query}, {\sf Answer}, \\ {\sf Reconstruct})$. 
	When the function class ${\mathcal F}$ supports weighted queries, this general syntax subsumes not only standard definitions of multi-server PIR, but also a broad class of predicate queries, making it expressive enough to capture most practical PIR use cases.
	This definition lays the groundwork for our constructions, which provide cryptographic guarantees of both privacy and publicly verifiable integrity in the multi-server setting.

	\begin{definition}
		A public verifiable private information retrieval scheme $\Pi$ for function class ${\mathcal F}\subseteq {\sf Funs}[[N]\times \{0,1\}^{\ell}, \mathbb{F}]$ 
		and database ${\bm X}=({\bm x}_1,\ldots,{\bm x}_N)$, weights ${\bm \omega}\in \mathbb{F}^N$, consists of four algorithms $\Pi = ({\sf KeyGen}, {\sf Query}, {\sf Answer}, {\sf Reconstruct})$ with the following syntax:
		\begin{itemize}
			\item $({\sf pk},{\sf sk}) \leftarrow {\sf KeyGen}(1^{\lambda})$: 
			The randomized key generation algorithm takes as input the security parameter $\lambda$.
			Output a public key ${\sf pk}$ and a secret key $\sf sk$.
			
			\vspace{1mm}
			\item $({\bm q}_1,{\bm q}_2,\ldots, {\bm q}_k,{\sf vk}) \leftarrow {\sf Query}({\sf pk}, {\sf sk}, f)$: 
			The randomized query algorithm takes the public key {\sf pk}, the secret key $\sf sk$, and a function $f\in \mathcal{F}$ as input. 
			Output $k$ queries $({\bm q}_1,{\bm q}_2,\ldots$, ${\bm q}_k)$ to servers 
			and publish a public verification key $\sf vk$.
			
			\vspace{1mm}
			\item ${\bm a}_j \leftarrow {\sf Answer}({\sf pk},{\bm X}, {\bm \omega}, {\bm q}_j)$: 
			The deterministic answer algorithm takes  
			the public key ${\sf pk}$, the database ${\bm X}$, the weights ${\bm \omega}$, and the query ${\bm q}_j$ as input.
			Output answer ${\bm a}_j$.
			
			\vspace{1mm}
			\item $m \leftarrow {\sf Reconstruct}({\bm a}_1,{\bm a}_2,\ldots, {\bm a}_k,{\sf pk},{\sf vk})$: 
			The deterministic reconstruct algorithm takes the public key $\sf pk$, the verification key ${\sf vk}$ and the answer $({\bm a}_1,{\bm a}_2,\ldots, {\bm a}_k)$  as input.
			Outputs an item $m\in \left\{\sum_{i\in [N]} {\bm \omega}_i\cdot f(i,{\bm x}_i),\, \bot \right\}$.
		\end{itemize}
	\end{definition}

	\begin{definition}
		\label{def:vpir}
		A publicly verifiable private information retrieval scheme $\Pi$ satisfies the following requirements:
		\begin{itemize}
			\item[$\bullet$] {\bf Correctness}. 
			Informally, the scheme $\Pi$ is correct, 
			when an honest client interacts with honest servers,  the client always recovers the weighted output of its chosen  function applied to the database, i.e., $\sum_{i\in [N]} {\bm \omega}_i f(i,{\bm x}_i)$.
			
			\vspace{1mm}
			\item[$\bullet$] {\bf Privacy}.	
			The scheme $\Pi$ is private, 
			when any coalition of up to $k-1$ servers cannot learn information about the user's input, 
			even if the servers have learned whether the client possesses a valid output or not. 
			Standard PIR schemes do not necessarily satisfy our strong notion of privacy, 
			since such schemes may be vulnerable to selective-failure attacks.
			
			\vspace{1mm}
			\item[$\bullet$] {\bf Security}. 
			The scheme $\Pi$ is secure, 
			when a client interacts with a set of $k$ servers, 
			where at most $k-1$ can be malicious and might collude with each other to persuade the client 
			with probability $1-{\sf negl}(\lambda)$.
		\end{itemize}
	\end{definition}
	We formally define these properties in Appendix \ref{app:pvpir}.

	\section{PVPIR Constructions}
	\label{sec:construction}
	In this section, we focus on the key construction of PVPIR schemes in the multi-server setting. 
	Our treatment is based on \textit{predicate query}, which serve as the general framework: the user retrieves all database entries that satisfy a specified predicate. 
	The classical \textit{point query} can be regarded as a special instance of this model, where the predicate simply selects a single item indexed by the query. 
	Thus, we first present constructions and analyses for predicate queries under different assumptions, and later explain how point query naturally follows as a direct specialization.

	\vspace{-1em}
	\subsection{Predicate Query}
	We consider the verifiable PIR schemes where the verification key {\sf vk} is public.	
	A central reason for this design choice is that public verifiability extends the trust model beyond the querying client. 
	While privately verifiable schemes restrict correctness checks to the client alone, public verifiability allows any external party to independently validate the server’s response. 
	This broader guarantee not only strengthens transparency but also makes the construction applicable in settings where multiple stakeholders must be assured of the query's integrity.

	Motivated by the applications of FSS in \cite{BGI15,BGI16}, this construction is parameterized by a number of servers $k\in \mathbb{N}$, 
	a database ${\bm X}=({\bm x}_1,\ldots,{\bm x}_N)$ of size $N\in \mathbb{N}$, with length of each entry being $\ell\in \mathbb{N}$, 	
	a function $f\in \mathcal{F}$, a weight vector ${\bm w}=({\bm \omega}_1,\ldots,{\bm \omega}_N)\in \mathbb{F}^N$, where $\mathcal{F}\subseteq {\sf Funs}[[N]\times \{0,1\}^{\ell}, \mathbb{F}]$ is a class of functions closed under scalar multiplication.
	
	\paragraph{\bf Construction 1.}
	\label{cons:predicateDl}
	The first publicly verifiable predicate query construction, $\Pi_2$, is illustrated in {Fig.}~\ref{construction: DL-based}.
	In the key generation phase $\sf KeyGen$, the client selects a group $\mathbb G$ of prime order $q_{\mathbb G}$ with a group generator $\xi$, and sets the public key as $(\mathbb G, q_{\mathbb G}, \xi)$.
	
	\vspace{1mm}
	In the $\sf Query$ phase, the client chooses a random
	$\alpha \in \mathbb{F}_{q_{\mathbb G}}\setminus \{0\}$ to compute the verification key ${\sf vk} = \xi^{\alpha}$.
	Then the client divides the function $f$ into $k$ additive function shares $(q_1,\ldots,q_k)$ such that $\sum_{j\in [k]} q_k=f$.
	For verification, choose the secret key $d$ to compute such a function $g$ that $g = \alpha \cdot f\in \mathcal{F}$, then split function $g$ into $k$ shares $(q'_1,\ldots,q'_k)$ such that $\sum_{j\in [k]} q'_k=g$. 
	The output of the $\sf Query$ phase is a pair of function shares ${\bm q}_j=(q_j, q'_j)$ for each server $j\in [k]$ and a public verification key ${\sf vk}$.
	Resulted from the DL assumption, given $\sf pk$ and $\sf vk$, 
	it is hard for a polynomial-time adversary to compute $\alpha$ from $\xi^{\alpha}$, if $q_{\mathbb G}$ is large enough. 
	
	\vspace{1mm}
	In the $\sf Answer$ phase, the $j$th server $\mathcal{S}_j$ evaluates the function shares on the database ${\bm X}$ and returns the answer ${\bm a}_j$ to the client. 
	By evaluating the function shares $(q_j,q'_j)$ at each database entry ${\bm x}_i,i\in [N]$, the server $j$ computes the answer ${\bm a}_j=(a_j,a'_j)$ as follows: 
	\begin{displaymath}
		a_j = \sum_{i\in[N]} {\bm \omega}_i \cdot q_j(i,\cdot) \in \mathbb{F}, \quad
		a'_j = \sum_{i\in[N]} {\bm \omega}_i \cdot q'_j(i,\cdot) \in \mathbb{F},
	\end{displaymath}
	where $q_j(i,\cdot)={\sf FSS.Eval}(j,q_j,i)$ 
	and $q'_j(i,\cdot)={\sf FSS.Eval}(j,q'_j,i)$.
	By the definition of $q_j$ and $q'_j$, we can see that 
	
	\begin{equation*}
		\label{eq:predicate1anser}
		\begin{aligned}
			\sum_{j\in [k]}a_j &= 
			\sum_{j\in [k]} \sum_{i\in[N]} {\bm \omega}_i \cdot q_j(i,\cdot) 
			&= 
			\sum_{i\in[N]}  f(i,\cdot) &= f({\bm X}), \\
			\sum_{j\in [k]}a'_j &= 
			\sum_{j\in [k]}\sum_{i\in[N]} {\bm \omega}_i \cdot q'_j(i,\cdot) &=
			\sum_{i\in[N]} g(i,\cdot) &= \alpha \cdot f({\bm X}).
		\end{aligned}
	\end{equation*}
	
	In the $\sf Reconstruct$ phase, the client reconstructs the answer  ${\bm a} = ({\bm a}^{(0)},{\bm a}^{(1)}) = 
	(\sum_{j\in [k]}a_j,\, \sum_{j\in [k]}a'_j)$. 
	By checking the equation 
	${\sf vk}^{({\bm a}^{(0)})} = \xi^{({\bm a}^{(1)})}$, 
	the client can verify the correctness of the answer. 
	If the verification passes, the client accepts the answer $m={\bm a}^{(0)}=\sum_{i\in[N]} f(i,\cdot)$ as the output of the query.
	
	\vspace{-1em}
	\begin{figure}[!ht]
		\centering
		\begin{boxedminipage}{12cm}
			{\bf Scheme $\Pi_1$ ($k$-server PVPIR for predicate queries to tolerating $k-1$ malicious servers, DL-based)}
			\begin{enumerate}
				\item  ${\sf KeyGen}(1^{\lambda}) \rightarrow ({\sf pk},{\sf sk})$: 
				\begin{enumerate}
					\item Chooses a prime order group $\mathbb{G}$, with the order $q_{\mathbb{G}}$ and a generator $\xi \in \mathbb{G}$.
					
					\vspace{1mm}
					\item Output the public key ${\sf pk}=(\mathbb{G}, q_{\mathbb{G}}, \xi)$, the secret key ${\sf sk}=\bot$.
				\end{enumerate}
				
				\vspace{1mm}
				\item {\sf Query}$({\sf pk},{\sf sk}, f) \rightarrow ({\bm q}_1,{\bm q}_2,\ldots, {\bm q}_k,{\sf vk}) $:
				\begin{enumerate}
					\item Sample a random field element $\alpha \in \mathbb{F}_{q_{\mathbb{G}}} \setminus \{0\}$. 
					
					\vspace{1mm}
					\item Set public verification key ${\sf vk}:= \xi^ \alpha$.					
					
					\vspace{1mm}
					\item Let $g=\alpha \cdot f$. Such a function $g$ exists according to the function class ${\mathcal F}$ is closed under scalar multiplication, as Definition \ref{def: closed under scalar multiplication}.
					
					\vspace{1mm}
					\item Compute $(q_1, q_2, \cdots, q_k) \leftarrow {\sf FSS.Gen}(1^{\lambda}, f)$ together with $(q'_1, q'_2, \cdots, q'_k) \leftarrow {\sf FSS.Gen}(1^{\lambda}, g)$.
					
					\vspace{1mm}
					\item Output $({\bm q}_1=(q_1, q'_1), \ldots, {\bm q}_k=(q_k, q'_k),{\sf vk})$.
				\end{enumerate}
				
				\vspace{1mm}
				\item ${\sf Answer}({\sf pk},{\bm X},{\bm \omega}, {\bm q}_j)\rightarrow {\bm a}_j$: 
				\begin{enumerate}
					
					\item Parse the query ${\bm q}$ as $(q_f, q_g)$, the database ${\bm X}$ as $({\bm x}_1, \ldots, {\bm x}_N)$ and the weights ${\bm \omega}=({\bm \omega}_1, \ldots, {\bm \omega}_N)$.
					
					\vspace{1mm}
					\item Evaluate function shares on the database and obtain $a_f \leftarrow \sum_{i\in[N]} {\bm \omega}_i \cdot {\sf FSS.Eval}(j,q_f, {\bm x}_i)$ and 
					$a_g \leftarrow \sum_{i\in[N]}{\bm \omega}_i \cdot {\sf FSS.Eval}(j,q_g, {\bm x}_i)$.

					\vspace{1mm}
					\item Output ${\bm a}_j \leftarrow (a_f, a_g)$.
				\end{enumerate}
				
				\vspace{1mm}
				\item ${\sf Reconstruct}({\bm a}_1,{\bm a}_2,\ldots, {\bm a}_k,{\sf pk},{\sf vk}) \rightarrow m$
				\begin{enumerate}
					\item Compute the sum ${\bm a}\leftarrow  {\bm a}_1 + {\bm a}_2 + \cdots + {\bm a}_k$.
					
					\vspace{1mm}
					\item Parse ${\bm a}$ as a pair $(m, \tau)$.
					
					\vspace{1mm}
					\item If ${\sf vk}^m=\xi^\tau$, output $m\in \mathbb{F}$, or $\perp$ otherwise.
				\end{enumerate}
				
			\end{enumerate}
			
		\end{boxedminipage}
		
		\caption{\small Scheme $\Pi_1$: $k$-server PVPIR for predicate queries tolerating $k-1$ malicious servers, based on DL assumption.}
		\label{construction: DL-based}	
	\end{figure}

	\begin{theorem}\label{thm:2}
		For the $k$-server verifiable PIR scheme  
		$\Pi_2 = ({\sf KeyGen}, {\sf Query}, {\sf Answer},\\ {\sf Reconstruct})$ 
		defined on function class ${\mathcal F}\subseteq {\sf Funs}[[N]\times \{0,1\}^{\ell},\, \mathbb{F}]$ 
		and database ${\bm X} = ({\bm {x}}_1,{\bm {x}}_2, \cdots, {\bm {x}}_N)\in \{0,1\}^{N\times \ell}$,
		given the function $f\in \mathcal{F}$, 
		$\Pi_2$ is correct, private, and secure as Definition \ref{def:vpir}.
	\end{theorem}

	\begin{proof}[Proof (Sketch)] 
		We prove the correctness, privacy and security of $\Pi_1$ in the following three parts.
		\begin{itemize}
			\item {\bf (Correctness)}: 
			The correctness of $\Pi_1$ follows from the correctness of the FSS scheme.
			Applying the correctness of FSS, $\sum_{j\in [k]} {\sf FSS.Eval}(j,f_j, i) = f(i,\cdot)$, 
			we obtain
			$$m = \sum_{j\in [k]}\sum_{i\in [N]}{\bm \omega}_i \cdot{\sf FSS.Eval}(j,f_j, i) = \sum_{i\in [N]}{\bm \omega}_i \cdot f(i,\cdot)= {\bm x}_{\iota},$$
			and similarly, $\tau=\alpha \cdot m$. 
			Therefore, ${\sf vk}^m=(\xi^\alpha)^m=\xi^\tau$.
			
			\vspace{1mm}
			\item {\bf (Privacy)}: 
			The privacy of $\Pi_1$ follows from the security of the FSS scheme.
			The proof is similar to that in \cite{BGI15,CNCWF23}.
			
			\vspace{1mm}
			\item {\bf (Security)}: 
			Consider a challenge ${\sf Exp}^{\rm Ver}_{{\mathcal A}, \Pi_1}(T, \lambda)$ adapted from {Fig.}~\ref{fig:sec of VPIR} between adversary ${\mathcal A}$ and its challenger, 
			we need to prove that 
			$$ \epsilon_{\Pi_2}(\lambda) = \Pr[{\sf Exp}^{\rm Ver}_{{\mathcal A}, \Pi_2}(T, \lambda)=1] $$ 
			is negligible. 
			WLOG, we assume that $T=[t]$.
			By the correctness of $\Pi_1$, the adversary ${\mathcal A}$ can always compute and output the correct 
			values $({\bf a}_1,\ldots,{\bf a}_t)$. 		
			Given $({\bf a}_1,\ldots,{\bf a}_t,{\bf a}_{t+1},\ldots, {\bf a}_k)$, 
			the algorithm ${\Pi_1}.\sf Reconstruct$ produces an output 
			$m={\bm x}_\iota$ and 
			verifies that 
			$\xi^{{\tau}}= (\xi^{\alpha})^{{m}}$.
			In contrast, within the experiment ${\sf Exp}^{\rm Ver}_{{\mathcal A}, \Pi_2}(T, \lambda)$, 
			a malicious adversary $\mathcal{A}$ typically outputs wrong values $(\check{{\bf a}}_1,\ldots,\check{{\bf a}}_t)$.

			Assume, for the sake of contradiction, that $\epsilon_{\Pi_2}(\lambda)$ is non-negligible.
			That is,  
			upon receiving $(\check{{\bf a}}_1,\ldots,\check{{\bf a}}_t,{\bf a}_{t+1},\ldots,{\bf a}_k)$, 
			the client runs ${\Pi_2}.\sf Reconstruct$ algorithm and successfully obtains  $(\check{m}, \check{\tau})$ such that $\check{m}\neq m$ 
			and $\xi^{\check{\tau}}= (\xi^{\alpha})^{\check{m}}$.
			This, however, implies an information leak, as the adversary can learn 
			\vspace{-1mm}
			$$\alpha = \frac{\tau-\check{\tau}}{m-\check{m}}, 
			{\rm \quad with\; } (m-\check{m},\tau-\check{\tau})=\sum_{j=1}^t ({\bf a}_j - \check{{\bf a}}_j).$$ 
			\vspace{-2em}
			
			\noindent
			Hence, as the adversary $\mathcal{A}$ already knows the public verification key $\xi^\alpha$, 
			it could potentially solve the discrete logarithm problem $(\xi,\xi^\alpha)$ by outputting the  correct $\alpha$ with non-negligible probability $\epsilon_{\Pi_2}(\lambda)$.
			This contradiction completes the proof of security for scheme $\Pi_1$ under the DL assumption. 
		\end{itemize}
	\end{proof}

	\vspace{-1em}
	\paragraph{\bf Construction 2.}
	The second construction, as presented in {Fig.}~\ref{construction: RSA-based}, is defined over a finite field $\mathbb{F}$.
	The database ${\bm X}=({\bm x}_1,\ldots,{\bm x}_N)$ consists of $N$ entries with each of length $\ell\in {\mathbb N}$.
	
	In the $\sf KeyGen$ phase, the client selects $n=pq$ with $p$ and $q$ being safe primes, samples $\xi \in {\mathbb Z}_n^{\times}$ such that $gcd(\xi,n)= 1$, and generates a conditional pair $(e,d)\mod \phi(n)$.
	The public key is set as $(\xi, n, e)$, and the secret key as $d$.

	In the $\sf Query$ phase, 
	the input function $f$ represents an aggregation operator over the database ${\bm X}$ matching a predicate, such as ${\sf sum}$ or ${\sf count}$.
	The client partitions $f$ into $k$ additive function shares $(q_1,\ldots,q_k)$ satisfying $\sum_{j\in [k]} q_k=f$.
	For verification, the client uses the secret key $d$ to compute a function $g$ that $g = d \cdot f\in \mathcal{F}$, which is then split into $k$ shares $(q'_1,\ldots,q'_k)$ such that $\sum_{j\in [k]} q'_k=g$. 
	
	The output of the $\sf Query$ phase consists of pairs of function shares ${\bm q}_j=(q_j, q'_j)$ for each server $j\in [k]$, along with a public verification key ${\sf vk}=e$.
	Under the RSA assumption, given $\sf pk$ and $\sf vk$, 
	it is computationally infeasible for a polynomial-time client to recover $d$ from the relation $\xi^{ed} \equiv 1 \mod n$, provided that $n$ is sufficiently large. 
	
	In the $\sf Answer$ phase, the $j$th server $\mathcal{S}_j$ evaluates the function shares on the database ${\bm X}$ and returns the corresponding answer ${\bm a}_j$ to the client. 
	Specifically, by evaluating the function shares $(q_j,q'_j)$ at each database entry ${\bm x}_i,i\in [N]$, server $j$ computes the answer pair ${\bm a}_j=(a_j,a'_j)$ as follows: 
	\begin{displaymath}
		a_j = \sum_{i\in[N]} {\bm \omega}_i \cdot q_j(i,\cdot) \in \mathbb{F}, 
		\quad
		a'_j = \prod_{i\in[N]} \xi^{{\bm \omega}_i \cdot q'_j(i,\cdot)} \in \mathbb{Z}_n^{\times},
	\end{displaymath}
	where $q_j(i,\cdot)={\sf FSS.Eval}(j,q_j,i)$ 
	and $q'_j(i,\cdot)={\sf FSS.Eval}(j,q'_j,i)$.
	Compared with the previous construction based on DL, the computations of $a'_j$ changes from addition to multiplication.
	By the definition of $q_j$ and $q'_j$, we can see that 
	\begin{equation*}
		\label{eq:predicate2answer}
		\begin{aligned}
			\sum_{j\in [k]}a_j &= 
			\sum_{j\in [k]} \sum_{i\in[N]} {\bm \omega}_i \cdot q_j(i,\cdot) 
			&= 
			\sum_{i\in[N]}  f(i,\cdot) &= f({\bm X}), \\
			\prod_{j\in [k]}a'_j &=  
			\prod_{j\in [k]}\prod_{i\in[N]} {\bm \omega}_i \cdot q'_j(i,\cdot) &=
			\prod_{i\in[N]} \xi^{g(i,\cdot)} &= \xi^{g({\bm X})} = \xi^{d \cdot f({\bm X})}.
		\end{aligned}
	\end{equation*}
	
	In the $\sf Reconstruct$ phase, the client reconstructs the final answer  
	${\bm a} = ({\bm a}^{(0)},{\bm a}^{(1)})$
	$$ {\bm a}^{(0)} = 
	\sum_{j\in [k]}a_j,\quad  {\bm a}^{(1)}=\prod_{j\in [k]}a'_j.$$ 
	
	\vspace{-2mm}
	The client then verifies correctness by checking whether 
	$\xi^{({\bm a}^{(0)})} = {\sf vk}^{({\bm a}^{(1)})}$, 
	if the equation holds, the client accepts the answer and outputs ${\bm a}^{(0)}=\sum_{i\in[N]} f(i,\cdot)$.
	
	\begin{figure}[!ht]
		\begin{center}
			\begin{boxedminipage}{11cm}
				{\bf Scheme $\Pi_2$ ($k$-server PVPIR for predicate queries to tolerating $k-1$ malicious servers, RSA-based)}
				\vspace{0.5em}
				\begin{enumerate}
					\item  ${\sf KeyGen}(1^{\lambda}) \rightarrow ({\sf pk},{\sf sk})$: 
					\begin{enumerate}
						\item Compute a RSA modulus $n=p\cdot q$, where $p,q$ are two safe primes. 
						
						\vspace{1mm}
						\item Sample a random  $\xi \in \mathbb{Z}_{n}^{\times}$ such that $gcd(\xi,n)= 1$. 
						
						\vspace{1mm}
						\item Sample a random value element $d \in \mathbb{Z}_n^{\times}$, 
						and compute $e\equiv d^{-1} \mod{\phi(n)}$, where $\phi(n)=(p-1)\cdot (q-1)$. 
						
						\vspace{1mm}
						\item Output the public key ${\sf pk}=(\xi,n,e)$ and the secret key ${\sf sk}=d$.
					\end{enumerate}
					
					\vspace{1mm}
					\item {\sf Query}$({\sf pk},{\sf sk}, f) \rightarrow ({\bm q}_1,{\bm q}_2,\ldots, {\bm q}_k,{\sf vk}) $:
					\begin{enumerate} 
						\vspace{1mm}
						\item Set public verification key ${\sf vk}:= e$.
						
						\vspace{1mm}					
						\item Let $g={\sf sk} \cdot f$. Such a function $g$ exists according to the function class ${\mathcal F}$ is closed under scalar multiplication, as Definition \ref{def: closed under scalar multiplication}.
						
						\vspace{1mm}
						\item Compute $(q_1, q_2, \cdots, q_k) \leftarrow {\sf FSS.Gen}(1^{\lambda}, f)$ together with $(q'_1, q'_2, \cdots, q'_k) \leftarrow {\sf FSS.Gen}(1^{\lambda}, g)$.
						
						\vspace{1mm}
						\item Output $({\bm q}_1=(q_1, q'_1), \ldots, {\bm q}_k=(q_k, q'_k),{\sf vk})$.
					\end{enumerate}
					
					\vspace{1mm}
					\item ${\sf Answer}({\sf pk}, {\bm X},{\bm \omega}, {\bm q}_j)\rightarrow {\bm a}_j$: 
					\begin{enumerate}
						\item Parse the query ${\bm q}$ as $(q_f, q_g)$, the database ${\bm X}$ as $({\bm x}_1, \ldots, {\bm x}_N)$ and the weights ${\bm \omega}=({\bm \omega}_1, \ldots, {\bm \omega}_N)$.
						
						\vspace{1mm}
						\item Evaluate function shares on the database and obtain $a_f \leftarrow \sum_{i\in[N]} {\bm \omega}_i \cdot {\sf FSS.Eval}(j,q_f, {\bm x}_i)$ and 
						$a_g \leftarrow \prod_{i\in[N]}\xi^{{\bm \omega}_i\cdot {\sf FSS.Eval}(j,q_g, {\bm x}_i)}$.
						
						\vspace{1mm}
						\item Output ${\bm a}_j \leftarrow (a_f, a_g)$.
					\end{enumerate}
					
					\vspace{1mm}
					\item ${\sf Reconstruct}({\bm a}_1,{\bm a}_2,\ldots, {\bm a}_k,{\sf pk},{\sf vk}) \rightarrow m$
					\begin{enumerate}
						\item 
						Reconstruct ${\bm a}\leftarrow ( {\bm a}^{(0)}_1  + \ldots + {\bm a}^{(0)}_k,\,\,
						{\bm a}^{(1)}_1 \cdot  \ldots \cdot {\bm a}^{(1)}_k )$.
						
						\vspace{1mm}
						\item Parse ${\bm a}$ as a pair $(m, \tau)$.
						
						\vspace{1mm}
						\item If $\xi^{m}\equiv \tau^{\sf vk} \mod{n}$, output $m\in \mathbb{F}$, or $\perp$ otherwise.
					\end{enumerate}
					
				\end{enumerate}
			\end{boxedminipage}
		\end{center}	
		\vspace{-1em}
		\caption{\small Scheme $\Pi_2$: $k$-server PVPIR for predicate queries tolerating $k-1$ malicious servers, based on RSA assumption.}
		\label{construction: RSA-based}	
		\vspace{-2em}
	\end{figure}

	\begin{theorem}
		For the $k$-server verifiable PIR scheme  $\Pi_3 = ({\sf KeyGen}, {\sf Query}, \\ {\sf Answer},  {\sf Reconstruct})$ defined on function class ${\mathcal F}\subseteq {\sf Funs}[[N]\times \{0,1\}^{\ell},\, \mathbb{F}]$ and database ${\bm X} = ({\bm {x}}_1,{\bm {x}}_2, \cdots, {\bm {x}}_N)\in \{0,1\}^{N\times \ell}$, given the function $f\in \mathcal{F}$, 
		$\Pi_3$ is correct, private, and secure as Definition \ref{def:vpir}.
	\end{theorem}
	
	\begin{proof}
		For \emph{correctness} and \emph{privacy}, the proof follows that in {Theorem}~\ref{thm:2}.
		we can reduce the security of $\Pi_3$ to the security of the RSA assumption.
		
		\vspace{2mm}
		For \emph{security}, considering a challenge ${\sf Exp}^{\rm Ver}_{{\mathcal A}, \Pi_3}(T, \lambda)$ adapted from {Fig.}~\ref{fig:sec of VPIR} between adversary ${\mathcal A}$ and its challenger,  
		we need to show 
		$$\epsilon_{\Pi_3}(\lambda) := \Pr[{\sf Exp}^{\rm Ver}_{{\mathcal A}, \Pi_2}(T, \lambda)=1] $$
		is negligible.
		Without loss of generality, we assume that $T=[t]$.
		Note that the correctness of $\Pi_3$ ensures that the adversary ${\mathcal A}$ can always compute and output correct 
		values $({\bf a}_1,\ldots,{\bf a}_t)$. 		
		Given $({\bf a}_1,\ldots,{\bf a}_t,{\bf a}_{t+1},\ldots, {\bf a}_k)$, 
		the ${\Pi_1}.\sf Reconstruct$ algorithm 
		outputs an $m=\sum_{i\in [N]} f(i,\cdot) = f({\bm X})$ and $\xi^{m}\equiv \tau^{\sf vk} \mod{n}$,
		whereas in the experiment ${\sf Exp}^{\rm Ver}_{{\mathcal A}, \Pi_3}(T, \lambda)$, 
		malicious adversary $\mathcal{A}$ usually outputs wrong answers $(\check{{\bf a}}_1,\ldots,\check{{\bf a}}_t)$.

		Assume for contradiction that $\epsilon_{\Pi_3}(\lambda)$is non-negligible.
		upon receiving $(\check{{\bf a}}_1,\ldots,\\ \check{{\bf a}}_t,{\bf a}_{t+1},\ldots,\check{{\bf a}}_k)$, 
		the client runs ${\Pi_3}.\sf Reconstruct$ algorithm and successfully outputs an $(\check{m}, \check{\tau})$ such that $\check{m}\neq m$ 
		and $\xi^{\check{m}}\equiv \check{\tau}^e \mod n$.
		This implies an information leak, since the adversary can learn such a $d$ that
		 $de\equiv 1 \mod \phi(n)$:
		 
		\vspace{-2mm}
		$$\xi^{de(m-\check{m})} \equiv \xi^{m-\check{m}} \mod n \quad \Rightarrow\quad
		de \equiv 1 \mod \phi(n).
		$$
		
		\vspace{-2mm}
		\noindent
		As the adversary $\mathcal{A}$ already knows the public verification key ${\sf vk} = e$, 
		$\mathcal{A}$ is able to 
		solve the RSA problem $((\xi^d)^e,\xi)$ by outputting correct $d$ with non-negligible probability $\epsilon_{\Pi_3}(\lambda)$.
		Therefore, the assumption is incorrect and the scheme $\Pi_3$ is $t$-secure under the RSA assumption.  
	\end{proof}

	\subsection{Point Query}
	We now consider point queries, which can be viewed as a special case of predicate queries.
	The protocol for point queries follows naturally from the general predicate query constructions.
	For clarity, we present the key parameters and protocol steps, omitting detailed proofs as they directly follow from the predicate query analysis.

	\vspace{-2mm}
	\paragraph{\bf Construction 3.} 
	The point query construction $\Pi_3$
	is defined over a finite field $\mathbb{F}$.	
	The database ${\bm X}=({\bm x}_1,\ldots,{\bm x}_N)$ is of size $N$, with length of each entry being $\ell\in {\mathbb N}$.
	With function class ${\mathcal F}$ being closed under scalar multiplication, we set ${\mathcal F}=\{f^{(1)},\ldots,f^{(N)}\} \subseteq {\sf Funs}[[N]\times \{0,1\}^{\ell}, \mathbb{F}]$, 
	$f^{(i)}(i,\cdot)=1$ and $f^{(i)}(i',\cdot)=0$ for all 
	$i'\neq i$. 
	The weights are the database entries themselves, 
	i.e., ${\bm \omega}_i={\bm x}_i \in \{0,1\}^\ell \subset \mathbb{F}$, for $i \in [N]$.

	In the $\sf KeyGen$ phase, ${\sf KeyGen}(1^{\lambda}) \rightarrow ({\sf pk},{\sf sk})$, 
	the client chooses a prime order group $\mathbb{G}$, with the order $q_{\mathbb{G}}$ and a generator $\xi \in \mathbb{G}$, the public key ${\sf pk}$ is set as $(\mathbb{G}, q_{\mathbb{G}}, \xi)$, 
	the secret key is set to be $\bot$. 
	
	In the $\sf Query$ phase, {\sf Query}$({\sf pk}, {\sf sk}, f) \rightarrow ({\bm q}_1,{\bm q}_2,\ldots, {\bm q}_k,{\sf vk}) $,
	the input function $f=f^{(\iota)}$ is chosen due to the index $\iota$ of the database entry ${\bm x}_\iota$ that the client wants to retrieve. 
	The client then splits the function $f^{(\iota)}$ into $k$ additive function shares $(q_1,\ldots,q_k)$ such that $\sum_{j\in [k]} q_j=f^{(\iota)}$. 
	For verification, choose a random non-zero field element $\alpha\in \mathbb{F}_{q_{\mathbb{G}}}$ and 
	compute a function $g^{(\iota)}$ such that $g^{(\iota)} = \alpha \cdot f^{(\iota)} \in \mathcal{F}$, 
	then split the $g^{(\iota)}$ into $k$ shares $(q'_1,\ldots,q'_k)$ such that $\sum_{j\in [k]} q'_j=g^{(\iota)}$. 
	The output of the $\sf Query$ phase is a pair of function shares ${\bm q}_j=(q_j, q'_j)$ for each server $j\in [k]$ and a public verification key ${\sf vk}=\xi^\alpha$.
	Thanks to the discrete logarithm (DL) assumption,
	given $\sf pk$ and $\sf vk$, 
	it is infeasible for anyone to compute $\alpha$ from $\xi^\alpha$, if $q_{\mathbb{G}}$ is large enough. 
	
	In the $\sf Answer$ phase, ${\sf Answer}({\sf pk},{\bm X},{\bm \omega}, {\bm q}_j)\rightarrow {\bm a}_j$,
	the $j$th server $\mathcal{S}_j$ evaluates the function shares on the database ${\bm X}$ and returns the answer ${\bm a}_j$ to the client. 
	By evaluating the function shares $(q_j,q'_j)$ at each database entry ${\bm X}_i,i\in [N]$, the server $j$ computes the answer ${\bm a}_j=(a_j,a'_j)$ as follows: 
	\begin{displaymath}
		a_j = \sum_{i\in[N]} {\bm \omega}_i \cdot q_j(i,\cdot) \in \mathbb{F}, \quad a'_j = \sum_{i\in[N]} {\bm \omega}_i \cdot q'_j(i,\cdot) \in \mathbb{F},
	\end{displaymath}
	where $q_j(i,\cdot)={\sf FSS.Eval}(j,q_j,i)$ 
	and $q'_j(i,\cdot)={\sf FSS.Eval}(j,q'_j,i)$.
	By the definition of $q_j$ and $q'_j$, we can see that 
	\begin{equation}
		\label{eq:pointanser}
		\begin{aligned}
			\sum_{j\in [k]}a_j &= \sum_{j\in [k]} \sum_{i\in[N]} {\bm \omega}_i \cdot q_j(i,\cdot) &= \sum_{i\in[N]} {\bm x}_i \cdot f^{(\iota)}(i,\cdot) &= {\bm x}_\iota, \\
			\sum_{j\in [k]}a'_j &= \sum_{j\in [k]}\sum_{i\in[N]} {\bm \omega}_i \cdot q'_j(i,\cdot) &= \sum_{i\in[N]} {\bm x}_i \cdot g^{(\iota)}(i,\cdot) &= \alpha \cdot {\bm x}_\iota.
		\end{aligned}
	\end{equation}
	
	In the $\sf Reconstruct$ phase, 
	${\sf Reconstruct}({\bm a}_1,{\bm a}_2,\ldots, {\bm a}_k,{\sf pk},{\sf vk}) \rightarrow m$,
	the client reconstructs the answer by adding the answers from all servers, i.e., ${\bm a} = ({\bm a}^{(0)},{\bm a}^{(1)}) = \sum_{j\in[k]} {\bm a}_j = (\sum_{j\in [k]}a_j,\sum_{j\in [k]}a'_j)$. 
	By checking the equation 
	${\sf vk}^{({\bm a}^{(0)})}=\xi^{({\bm a}^{(1)})}$, 
	the client can verify the correctness of the answer. 
	If the verification passes, the client accepts the answer $m={\bm a}^{(0)}=\sum_{i\in[N]} {\bm x}_i \cdot f^{(\iota)}(i,\cdot)$ as the output of the query.
	
	\begin{theorem}\label{thm:point-query}
		For the $k$-server verifiable PIR scheme  
		$\Pi_3 = ({\sf KeyGen}, {\sf Query}, {\sf Answer},\\ {\sf Reconstruction})$ 
		defined on function class ${\mathcal F}\subseteq {\sf Funs}[[N]\times \{0,1\}^{\ell},\, \mathbb{F}]$ 
		and database ${\bm X} = ({\bm {x}}_1,{\bm {x}}_2, \ldots, {\bm {x}}_N)\in \{0,1\}^{N\times \ell}$,
		given the queried index is $\iota\in [N]$, weights vector ${\bm \omega}={\bm X}$, and the function $f\in \mathcal{F}$ defined as $f(\iota,\cdot)= 1$ and $f(x,\cdot)= 0$ if $x \neq \iota$, 
		$\Pi_3$ is correct, private, and secure as Definition \ref{def:vpir}.
	\end{theorem}
	
	\vspace{-1mm}
	The proof is straightforwards and closely parallels that of Theorem~\ref{thm:2}, we therefore leave it to the reader.
	
	\vspace{-2mm}
	\section{Performance Evaluation}
	\label{sec:Implementation}
	In this section, we experimentally evaluate all the three verifiable PIR protocols proposed in \Cref*{sec:construction},
	all experiments are conducted under the two-server setting from the literature.
	Although the protocol is formally defined for the general multi-server setting, our experiments focuses on the two-server case. In what follows, \textit{multi-server} refers to the defined model, while \textit{two-server} denotes our concrete realization.
	The chosen parameters ensure that each protocol achieve a security level of 124 bits.

	{\bf Parameters.} 
	The computational performance of these protocols is measured on the Ubuntu 20.04 with 16GB unified memory and 6-cores 4.0GHz Intel Xeon E5-2286G CPU made by GenuineIntel.
	All codes are executed using GO 1.24.1.

	We instantiate our publicly verifiable PIR protocol  $\Pi_1$
	and $\Pi_2$ for predicate query and $\Pi_3$.
	The message field is chosen to be $\mathbb{F}^4_p$ with $p=2^{32}-1$.
	In $\Pi_1$ scheme, based on the DL assumption,  
	we choose the cyclic group $\mathbb{G}$ of order being $p-1$, where $p$ is a safe parameter of 3072 bits.
	In $\Pi_2$ scheme, based on the RSA assumption, the module $n=pq$ is chosen to be of 3072 bits.
	
	{\bf Experimental methodology.} 
	Our evaluation focuses on key performance metrics, including user and server computation time, communication overhead, and scalability with respect to database size. 
	Each protocol is executed 50 times with randomly chosen queries, and averaged results are reported to ensure statistical reliability. 
	User computation time includes both query generation and reconstruction, while communication cost covers query upload and answer download. All experiments are conducted on the same physical machine under consistent settings, though minor system-level fluctuations (e.g., background processes) may cause negligible timing variance.

	\subsection{Multi-server predicate query}	

	Unlike prior two-server VPIR schemes for predicate queries~\cite{CNCWF23}, which provide only private verifiability, our constructions achieve public verifiability through lightweight public-key operations, enabling third-party verification while keeping the cost—primarily group exponentiations practical.

	\vspace{-6mm}
	\begin{figure}	
		\centering
		\includegraphics[scale=0.6]{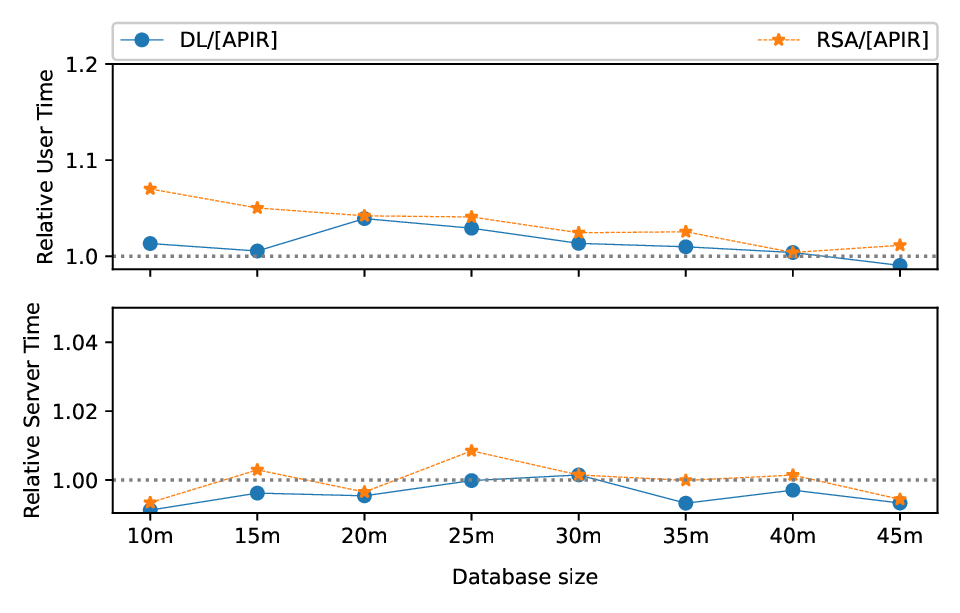}
		\vspace{-3mm}
		\caption{\small 
			The relative user time and server time comparison between DL based scheme and APIR\cite{CNCWF23}, RSA based scheme and APIR\cite{CNCWF23} in two-server setting. } 
		\label{perf:dlp_rsa_vs_fss}
	\end{figure}
	\vspace{-3em}
	
	Experimental results for both schemes $\Pi_1$ and $\Pi_2$ show that the added computation increases user time, introduced by public-key operations, by at most 1.1$\times$ and 1.05$\times$, and server time by only 1.02$\times$. 
	As illustrated in {Fig.}~\ref{perf:dlp_rsa_vs_fss}, across database sizes from 10 to 45 million items, the overhead remains comparable to privately verifiable baselines, confirming that public verifiability can be realized with minimal performance impact.

	\subsection{Multi-server Point Query}		
	From a theoretical standpoint, the communication complexity of the two-server scheme $\Pi_3$ grows only mildly with the database size $N$. In contrast, the Merkle tree–based scheme\cite{CNCWF23} incurs steadily increasing bandwidth overhead, as each response must be accompanied by a Merkle proof. Furthermore, the augmented database grows super-linearly relative to the original, imposing significant server-side memory costs (see Appendix~\ref{fig:Database size of MerkleTree}).
	
	\vspace{-2em}
	\begin{figure}[!htp]
		\centering
		\includegraphics[scale=0.65]{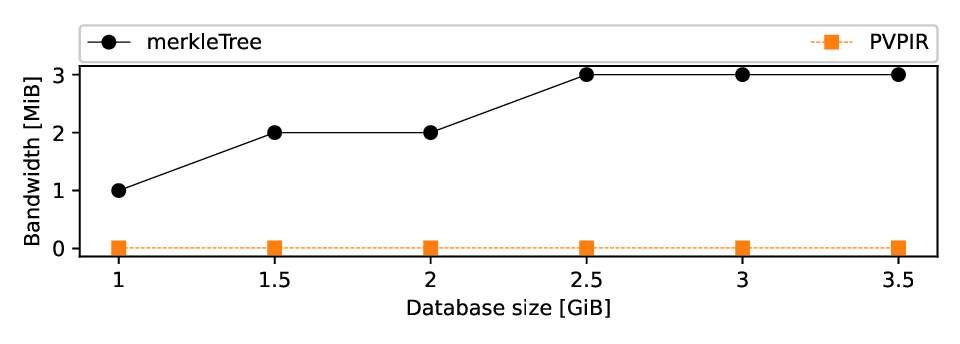}
		\vspace{-1em}
		\caption{\small
			Bandwidth comparison between  Merkle tree-based VPIR\cite{CNCWF23} and PVPIR $\Pi_3$ for point query.
			The overhead of retrieving a 256-Byte data item applying the Merkle tree-based VPIR scheme\cite{CNCWF23} and the our VPIR scheme $\Pi_3$ in two-server setting.
		} 
		\label{perf:FSS-vs-MerkleTree}
	\end{figure}
	
	\vspace{-3em}
	\begin{figure}[!h]
		\centering
		\vspace{-1em}
		\includegraphics[scale=0.6]{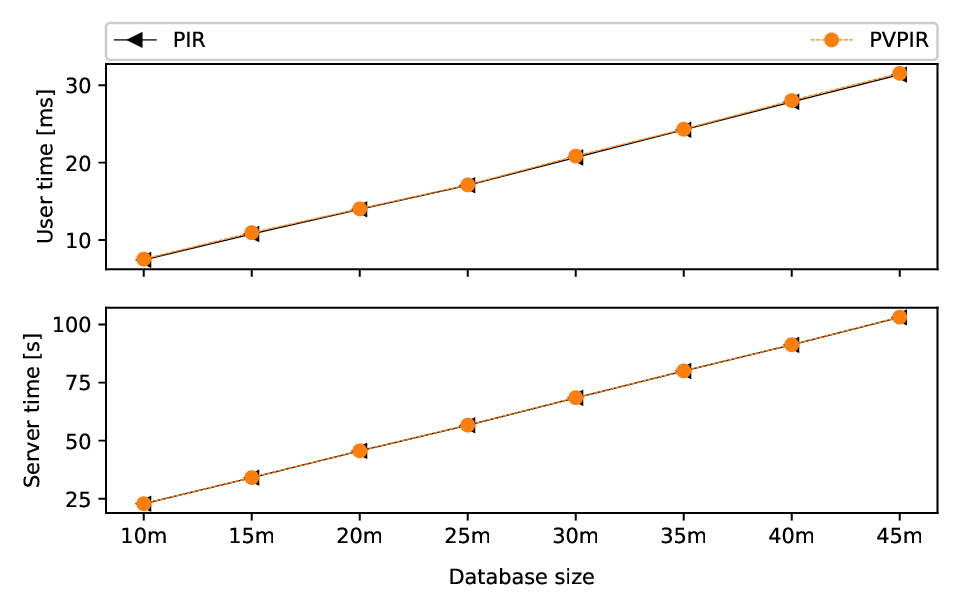}
		\vspace{-1em}
		\caption{\small Time consumption comparison between PVPIR $\Pi_3$ scheme and classic PIR scheme in two-server setting.
		Here the size of database is measured by the amount of database entries and `{\sf m}' means `million'.}
		\label{perf:pir_vs_vpir}
	\end{figure}

	In $\Pi_3$, the query and response sizes are $\mathcal{O}(\lambda 2^{\ell/2})$ and $\mathcal{O}(2\ell)$, respectively, where $\ell$ denotes the bit-length of each data item. Both quantities are independent of $N$, highlighting the scalability and communication efficiency of our construction.
	{Fig.}~\ref{perf:FSS-vs-MerkleTree} compares the bandwidth overhead of our verifiable PIR scheme $\Pi_3$ with the Merkle tree–based approach for point queries. As shown, $\Pi_3$ incurs significantly lower communication cost, especially as the database grows, demonstrating superior scalability. This efficiency arises because each point query is mapped to a fixed-size hidden predicate query, making the bandwidth independent of $N$.		
	
	The dominant computational overhead in $\Pi_3$ arises from evaluating the verification function $g$, which ensures the correctness of the retrieved result. As illustrated in {Fig.}~\ref{perf:pir_vs_vpir}, when the database size increases from 10 million to 45 million items, the additional cost of verifiability remains minimal. 
	This shows that the overhead of integrating verification into the FSS-based scheme is marginal, confirming its practical efficiency.

	\section{Conclusion}
	\label{sec:conclusion}
	In this paper, we have studied publicly verifiable private information retrieval (PVPIR) protocols for predicate query and point query in the multi-server setting, achieving both query privacy and result correctness. 
	We have constructed efficient PVPIR protocols based on DL and RSA assumptions for predicate query.
	For the point query, our construction introduces a lightweight verification mechanism with minimal computational overhead and significantly lower communication cost compared to existing Merkle tree–based solutions. 
	Both theoretical analysis and empirical evaluation confirm that our protocol sustains stable bandwidth usage and strong scalability, making it well-suited for practical deployment in large-scale privacy-preserving data retrieval.

	While our constructions provide efficient publicly verifiable PIR for predicate and point queries, several directions remain open. Future work includes improving efficiency for very large databases, supporting dynamic updates, and enhancing robustness against fully malicious or adaptive adversaries.
	
	Another promising direction is exploring hybrid verifiability schemes that balance public and private auditing, as well as applying PVPIR protocols to real-world systems such as federated analytics, blockchain-based storage, and secure multi-party computation. These efforts can further enhance the practicality, flexibility, and security of publicly verifiable PIR.

	\begin{credits}
		\subsubsection{\ackname}
			This work is supported by the Shenzhen-Hong Kong Joint Funding Project (Category A), under grant No. SGDX20240115103359001.
	\end{credits}

	\newpage
	%
	
	\bibliographystyle{Style/splncs04}
	\bibliography{Bib/vpir2025}

	\begin{subappendices}
		\renewcommand{\thesection}{\Alph{section}}		
		
		\section{Assumptions for PVPIR Constructions}
		\vspace{-1mm}
		\begin{definition}{\bf (DL assumption)} 
			Let \( \mathbb{G} \) be a cyclic group of prime order \( q \) generated by an element \( g \). 
			The \emph{Discrete Logarithm (DL) Assumption} states that, given \( g \in \mathbb{G} \) and \( h = g^x \in \mathbb{G} \) for a uniformly random \( x \in \mathbb{Z}_q \), 
			it is computationally infeasible for any probabilistic polynomial-time (PPT) adversary \( \mathcal{A} \) to output \( x \). 
			
			Formally, for any PPT algorithm \( \mathcal{A} \), the advantage
			$\Pr\left[ \mathcal{A}(g, g^x) = x \right]$
			is negligible in the security parameter \( \lambda \), where the group \( \mathbb{G} \) is generated by a group generation algorithm \( \mathcal{G}(1^\lambda) \).
			
		\end{definition}
		
		\begin{definition}{\bf (RSA assumption)}
			Let \( \mathsf{RSA}_{\lambda} \) be an algorithm that outputs an RSA modulus \( n = pq \) for two random \( \lambda \)-bit primes \( p, q \), 
			and let \( e \in \mathbb{Z}_n^{\times} \) be a public exponent such that \( \gcd(e, \phi(n)) = 1 \). 
			The \emph{RSA Assumption} asserts that, given \( n, e \), and a random \( y \in \mathbb{Z}_n^{\times} \), 
			it is computationally infeasible for any PPT adversary \( \mathcal{A} \) to compute \( x \in \mathbb{Z}_n^{\times} \) such that \( x^e \equiv y \pmod{n} \).
			
			Formally, for any PPT adversary \( \mathcal{A} \), the success probability
			\[
			\Pr\left[ \mathcal{A}(n, e, y) = x \text{ such that } x^e \equiv y \pmod{n} \right] \leq {\sf negl}(\lambda).
			\]
		\end{definition}

		\section{Multi-server PVPIR Definitions}
		\vspace{-1mm}
		\label{app:pvpir}
		\begin{definition}
			A multi-server publicly verifiable PIR scheme $\Pi = ({\sf KeyGen}, {\sf Query}$, ${\sf Answer}, {\sf Reconstruct})$ is correct if for all function class $\mathcal{F}\subseteq {\sf Funs}[[N]\times\{0,1\}^\ell,\mathbb{F}]$, database $\bm X=({\bm x}_1,\ldots, {\bm x}_N)\in \{0,1\}^{N\times \ell}$, weights ${\bm \omega}=({\bm \omega}_1,\ldots, {\bm \omega}_N)\in \mathbb{F}^N$, and $\lambda \in \mathbb{N}, f\in \mathcal{F}$, satisfies the following requirements:
			\begin{itemize}
				\vspace{1mm}
				\item[$\bullet$] {\bf Correctness}. 
				The correctness requires that the reconstruction algorithm outputs the correct result if all algorithms are honestly executed. That is, 
				\begin{displaymath}
					\Pr\left[
					\begin{matrix}
						({\sf pk},{\sf sk}) \leftarrow {\sf KeyGen}(1^{\lambda}) \\
						({\bm q}_1,{\bm q}_2,\ldots, {\bm q}_k,{\sf vk}) \leftarrow {\sf Query}({\sf pk},{\sf sk}, f) \\ 
						({\bm a}_1,{\bm a}_2,\ldots, {\bm a}_k) \leftarrow {\sf Answer}({\sf pk}, {\bm X},{\bm \omega}, {\bm q}_j) \\
						m\leftarrow {\sf Reconstruct}({\bm a}_1,{\bm a}_2,\ldots, {\bm a}_k,{\sf pk},{\sf vk}): \\
						m\in \{\sum_{i\in [N]} {\bm \omega}_i \cdot f({\bm x}_i), \perp\}
					\end{matrix}\right] = 1
				\end{displaymath}
				where the probability is computed over all the random coins used by the algorithms of the scheme.
				
				\vspace{1mm}
				\item[$\bullet$] {\bf Privacy}.	
				Our scheme is private against the selective failure attack, formally, for any adversary ${\mathcal A}=({\mathcal A}_0,{\mathcal A}_1)$ who controls $k-1$ servers, 
				and an honest server whose index is $j_{\sf good}\in [k]$, we define the distribution 
				\begin{displaymath}
					\mathsf{REAL}_{\mathcal{A}, j_{\mathsf{good}}, f, \lambda, {\bm X}, {\bm w}} = 
					\left\{
					\hat{\beta} : 
					\begin{aligned}
						&({\sf pk},{\sf sk})\leftarrow {\sf KeyGen}(1^\lambda)\\
						&({\bm q}_1, \ldots, {\bm q}_k,{\sf vk}) \leftarrow \mathsf{Query}({\sf pk},{\sf sk}, f) \\
						&{\bm a}_{j_{\mathsf{good}}} \leftarrow \mathsf{Answer}({\sf pk}, {\bm X},{\bm \omega}, {\bm q}_{j_{\mathsf{good}}}) \\
						&({\sf st}_\mathcal{A},\{{\bm a}_j\}_{j \ne j_{\mathsf{good}}}) \leftarrow \mathcal{A}_0({\sf pk},{\bm X},{\bm w}, \{{\bm q}_j\}_{j \ne j_{\mathsf{good}}}) \\
						&y \leftarrow \mathsf{Reconstruct}({\bm a}_1, \ldots, {\bm a}_k,{\sf vk}) \\
						&b = 1 \text{ if } y \ne \bot \text{ else } 0\\
						&\hat{\beta} \leftarrow \mathcal{A}_1({\sf st}_\mathcal{A},b)
					\end{aligned}
					\right\}.
				\end{displaymath}
				Similarly, for a simulator ${\bm S}=({\bm S}_0,{\bm S}_1)$, 
				define the distribution 
				\begin{displaymath}
					\mathsf{IDEAL}_{\mathcal{A}, {\bm S},\mathcal{F}, \lambda, {\bm X}, {\bm w}} = 
					\left\{
					\beta : 
					\begin{aligned}
						&({\sf pk},{\sf sk})\leftarrow {\sf KeyGen}(1^\lambda)\\
						&({\bm q}_1, \ldots, {\bm q}_k,{\sf st}_{\bm S}) \leftarrow {\bm S}_0({\sf pk},{\sf sk}, \mathcal{F},{\bm X},{\bm \omega}) \\
						&({\bm a}_1, \ldots, {\bm a}_k,{\sf st}_{\mathcal{A}})\leftarrow \mathcal{A}_0({\sf pk}, {\bm X},{\bm \omega}, {\bm q}_{j_{\mathsf{good}}}) \\
						&b \leftarrow {\bm S}_1({\bm a}_1, \ldots, {\bm a}_k,{\sf st}_{\bm S}) \\
						&\beta \leftarrow \mathcal{A}_1({\bm a}_1, \ldots, {\bm a}_k,{\sf st}_\mathcal{A})
					\end{aligned}
					\right\}.
				\end{displaymath}
				We say $\Pi$ is private if for every efficient adversary ${\mathcal A}=(\mathcal{A}_0,\mathcal{A}_1)$, 
				any database ${\bm X}$, and weights ${\bm \omega}$, 
				there exists a simulator ${\bm S}=({\bm S}_0,{\bm S}_1)$ such that for all $\lambda \in \mathbb{N}$, $f\in \mathcal{F}$, 
				$j_{\sf good}\in [k]$, the following holds: 
				\begin{displaymath}
					\mathsf{REAL}_{\mathcal{A}, j_{\mathsf{good}}, f, \lambda, {\bm X}, {\bm w}} \approx_c \mathsf{IDEAL}_{\mathcal{A}, {\bm S},\mathcal{F}, \lambda, {\bm X}, {\bm w}} = 1.
				\end{displaymath}
				\item[$\bullet$] {\bf Security}.
				Consider the experiment 
				${\sf Exp}^{\rm Ver}_{{\mathcal A}, \Pi}(T, \lambda)$
				of challenge for corrupted parties $T \subset [k]$ of cardinality $t<k$,	
				for all probabilistic polynomial time (PPT) adversary ${\mathcal A}$, $\Pr[{\sf Exp}^{\rm Ver}_{{\mathcal A}, \Pi}(T, \lambda)=1] \leq {\sf negl}(\lambda)$.
			\end{itemize}
			
		\end{definition}
		
		\vspace{-2em}
		\begin{figure}[htbp]
			\begin{center}
				\begin{boxedminipage}{11cm}
					\begin{enumerate}
						\item[{\rm (a)}] The adversary ${\mathcal A}$ choose a database ${\bm X}$ and query function $f\in \mathcal{F}$.
						
						\vspace{1mm}
						\item[{\rm (b)}] The challenger executes $({\sf pk}, {\sf sk}) \leftarrow {\sf KeyGen}(1^\lambda)$,
						$({\bm q}_1,{\bm q}_2,\ldots, {\bm q}_k, {\sf vk}) \leftarrow {\sf Query}({\sf pk},{\sf sk}, f)$ 
						and ${\bm a}_{j} \leftarrow {\sf Answer}({\sf pk},{\bm X},{\bm q}_{j})$ for $j \notin T$.
						Then it sends queries $\{{\bm q}_j\}_{j\notin T}$ to ${\mathcal A}$.
						
						\vspace{1mm}
						\item[{\rm (c)}] The adversary ${\mathcal A}$ chooses $t$ random  $\{\check{\bm a}_j\}_{j\in T}$ and send to the challenger.
						
						\vspace{1mm}
						\item[{\rm (d)}] The challenger executes the reconstruction algorithm $m \leftarrow{\sf Reconstruct}(\{\check{\bm a}_j\}_{j\in T},\{{\bm a}_j\}_{j\notin T}.{\sf vk})$.
						
						\vspace{1mm}
						\item[{\rm (e)}] If $m\not\in\{\sum_{i\in[N]}{\bm \omega}_i f(i,{\bm X}_i), \perp\}$, output 1, otherwise 0.
					\end{enumerate}
				\end{boxedminipage}
				\caption{\small  ${\sf EXP}_{\mathcal{A},\Gamma}^{\rm Ver}(T, \lambda)$}
				\label{fig:sec of VPIR}
			\end{center}	
			\vspace{-2em}		
		\end{figure}

		\vspace{-6mm}
		\section{Database Size of Merkle tree-based VPIR}
		\vspace{-1mm}
		\label{fig:Database size of MerkleTree}
		We use the Merkle tree–based VPIR scheme as the baseline for point query. Notably, its database size increases super-linearly with the original dataset as in {Fig.}~\ref{fig:MerkleTree database size}, which imposes significant memory pressure on the server.
		
		\vspace{-2em}
		\begin{figure}[h]
			\centering
			\includegraphics[scale=0.65]{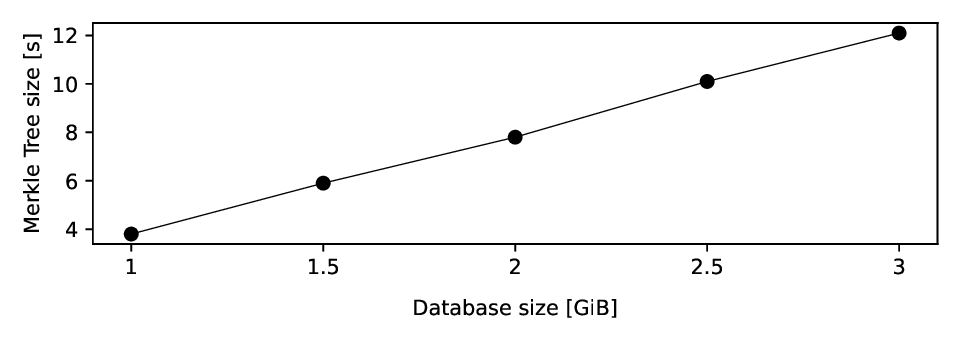}
			\vspace{-1em}
			\caption{\small Merkle tree database size} 
			\label{fig:MerkleTree database size}
		\end{figure}

	\end{subappendices}
	
\end{document}